\documentclass[a4paper,USenglish,numberwithinsect]{lipics-v2021}

\hideLIPIcs
\nolinenumbers

\pdfoutput=1
\usepackage[utf8]{inputenc}
\usepackage[T1]{fontenc}
\usepackage{amsmath,amsthm,amssymb,thmtools,thm-restate,booktabs,color,doi,graphicx,latexsym,url,xcolor,xspace,soul}
\usepackage[numbers,sort&compress]{natbib}
\usepackage{microtype,hyperref}
\usepackage[inline]{enumitem}
\setlist[itemize]{label=--}
\setlist[enumerate]{label=(\arabic*),labelindent=\parindent,leftmargin=*}

\usepackage[ruled, lined, linesnumbered, commentsnumbered, noend, ]{algorithm2e}
\usepackage[noend]{algpseudocode}

\usepackage{enumitem}
\usepackage{sectsty}
\usepackage{todonotes}
\allsectionsfont{\boldmath}

\definecolor{citecolor}{HTML}{0000C0}
\definecolor{urlcolor}{HTML}{000080}

\hypersetup{
    colorlinks=true,
    linkcolor=black,
    citecolor=citecolor,
    filecolor=black,
    urlcolor=urlcolor,
    pdftitle={Online List Access with Precedence Constraints}
}

\newtheorem{assumption}[theorem]{Assumption}

\newcommand{\OPT}{\textsc{OPT}\xspace}

\newcommand{\RALG}{\textsc{RAND}\xspace}
\newcommand{\RAND}{\textsc{RAND}\xspace}
\newcommand{\DET}{\textsc{DET}\xspace}

\newcommand{\pos}{\textsf{pos}}
\newcommand{\dep}{\ensuremath{\textbf{d}}}
\newcommand{\depx}{\ensuremath{\delta}}
\newcommand{\bit}{\ensuremath{\textbf{b}}}

\newcommand{\Id}{\ensuremath{D}}
\newcommand{\Ic}{\ensuremath{A}}

\newcommand{\DAG}{\textsc{G}\xspace}

\newcommand{\MRF}{\textsc{MRF}\xspace}
\newcommand{\MRFx}[1]{\textsc{MRF}(#1)\xspace}

\title{Online List Access with Precedence Constraints}

\author{Maciej Pacut}{Faculty of Computer Science, University of Vienna, Austria}{maciej.pacut@univie.ac.at}{https://orcid.org/0000-0002-6379-1490}{}
\author{Juan Vanerio}{Faculty of Computer Science, University of Vienna, Austria}{juan.vanerio@univie.ac.at}{https://orcid.org/0000-0003-3120-5028}{}
\author{Vamsi Addanki}{Faculty of Computer Science, University of Vienna, Austria}{vamsi.addanki@univie.ac.at}{https://orcid.org/0000-0002-0577-0413}{}
\author{Arash Pourdamghani}{Faculty of Computer Science, University of Vienna, Austria}{arash.pourdamghani@univie.ac.at}{https://orcid.org/0000-0002-9213-1512}{}
\author{Gabor Retvari}{Department of Telecommunications and Media Informatics, Budapest University of Technology and Economics
Budapest, Hungary}{retvari@tmit.bme.hu}{}{}
\author{Stefan Schmid}{Faculty of Computer Science, University of Vienna, Austria}{stefan.schmid@univie.ac.at}{https://orcid.org/0000-0002-7798-1711}{}

\authorrunning{M.~Pacut, J.~Vanerio, V.~Addanki, A.~Pourdamghani, G.~Retvari, S.~Schmid}


\ArticleNo{\mbox{}}

\begin{document}

\maketitle

\begin{abstract}
This paper considers a natural generalization of the online list access problem in the paid exchange model, where additionally there can be
precedence constraints (``dependencies'') among the nodes in the list.
For example, this generalization is motivated by applications in the context of packet classification.
Our main contributions are constant-competitive deterministic and randomized online algorithms, designed around a procedure Move-Recursively-Forward, a generalization of Move-To-Front tailored to handle node dependencies.
Parts of the analysis build upon ideas of the classic online algorithms Move-To-Front and BIT, and address the challenges of the extended model.
We~further discuss the challenges related to insertions and deletions.
\end{abstract}

\thispagestyle{empty}
\setcounter{page}{0}

\section{Introduction}

The list access problem is a fundamental online algorithmic problem, originally introduced
by Sleator and Tarjan in 1985 in their seminal work on amortized analysis~\cite{sleator1985amortized}.
In a nutshell, in the list update problem, we consider a linked list data structure in which accessing
a~node costs proportionally to its distance from the head of the list. We are looking for an algorithm that adaptively
reorders the list in order to minimize the access cost, when faced with a~sequence of access requests 
revealed in an online manner by an adversary. As usual in competitive analysis, we are interested in the competitive ratio:
we compare the overall cost of an online algorithm to the cost of an optimal offline algorithm which knows the access
sequence ahead of time.

This paper initiates the study of a generalization of this classic problem where there can be
precedence constraints among the nodes in the list. These precedence constraints define a~relative order
among some nodes in the list, which means that some nodes must be placed before others.
We assume that the partial order is given in form of a directed acyclic graph that induces this order.
This generalization is natural
and motivated by practical applications, e.g., in the context of packet classification in communication
networks. For example, rules in an IP router or in a firewall are often ordered to ensure a correct handling.
A detailed motivating example will follow below.

\pagebreak
The consideration of constraints poses an algorithmic challenge. In contrast to existing online algorithms for the list access problem,
such as Move-to-Front and BIT which upon access, can flexibly move nodes towards the head of the list,
constraints may prevent such optimizations or at least make them costly.
The constraints raise the question
whether in order to move forward an accessed node $u$, it may be worthwhile to recursively move forward the nodes $v$ on which
$u$ depends as well, and to which extent.
In particular, it is easy to see that if the directed acyclic graph
describing the precedence constraints has depth $d$, then a~strategy that aggressively moves nodes forward is at best $\Omega(d)$-competitive.

Furthermore, we consider insertions and deletions, a particularly challenging, yet integral part of the problem.
Although insertions and deletions were handled in the original work on list update~\cite{sleator1985amortized}, many papers considering the paid exchange model do not address this issue, and requests are restricted to accesses only~\cite{Albers03,KamaliSurvey2013,Lopez-OrtizRR20}.
We start with results for the access-only variant of the problem, and then we consider insertions and deletions, and discuss the assumptions needed to handle them efficiently.

Our main contribution in this paper are two constant-competitive online algorithms
for the list access problem with precedence constraints, a deterministic one and a randomized one.
Before we present our contributions in detail, we present our model. Subsequently, we give a practical motivation for the problem and then put our contribution into perspective with regard to related work.

\subsection{Preliminaries and Competitive Ratio}
\label{sec:online}

The sequence of requests $\sigma$ is revealed one-by-one, in an online fashion. Upon seeing a~request, the algorithm must serve it without the knowledge of future requests.
We measure
the performance of an online algorithm by comparing to the performance of an optimal offline
algorithm. Formally, let~$\DET(\sigma)$, resp.~$\OPT(\sigma)$, be the cost
incurred by a deterministic online algorithm \DET, resp.~by an optimal offline
algorithm \OPT, for a given sequence of requests~$\sigma$. In contrast to \DET, which learns the~requests one-by-one as
it serves them, \OPT has complete knowledge of the entire request
sequence~$\sigma$ \emph{ahead of~time}.
The goal is to design online
algorithms that provide worst-case guarantees. In particular, $\DET$ is said
to be \emph{$\alpha$-competitive} if there is a~constant~$\beta$, such that for any
input sequence~$\sigma$ it holds that
\[
	\DET(\sigma) \leq \alpha \cdot \OPT(\sigma) + \beta.
\]
Note that $\beta$ cannot depend on input $\sigma$ but can depend on other
parameters of the problem, such as the number of nodes.
The minimum $\alpha$ for which $\DET$ is $\alpha$-competitive is called the
\emph{competitive ratio} of $\DET$.
We say that $\DET$ is \emph{strictly $\alpha$-competitive} if additionally $\beta=0$.

We say that a randomized online algorithm $\RALG$ is \emph{$\alpha$-competitive} if
\[
E[\RALG(\sigma)] \le \alpha \cdot \OPT(\sigma)+\beta
\]
for any input sequence $\sigma$ and a fixed constant $\beta$.
The input sequence and the benchmark solution $\OPT$ is generated by an adversary.
We distinguish between the notion of competitiveness against various adversaries, having different knowledge about $\RALG$ and different knowledge while producing the offline benchmark solution $\OPT$.
Competitive ratios for a~given problem may vary depending on the power of the adversary.
In our work, we design algorithms against an~\emph{oblivious offline adversary} that must produce an input sequence in advance, merely knowing the description of the algorithm it competes against (in particular, it may have access to probability distributions that the algorithm uses, but not the random outcomes), and pays an optimal offline cost for the sequence.
For a~comprehensive overview of adversary types, see \cite{Borodin1998}.

\subsection{Model}
\label{sec:model}

Our task is to manage a self-adjusting linked list serving a sequence of requests, with minimal access and rearrangement costs
and accounting for precedence constraints induced by a~directed acyclic graph~\DAG.
If there are no constraints, the problem is equivalent to the classic list access problem.

\vspace{-0.5cm}
\paragraph*{The list and the requests.}
Consider a set of $n$ nodes arranged in a linked list.
Over time, we receive requests from a sequence $\sigma$, describing \emph{accesses} to nodes from our list\footnote{We refer to Section~\ref{sec:insertions} for a discussion about node insertions and deletions.}.
Upon receiving an access request to a node in the list, an algorithm searches linearly through the list,
starting from the head of the list, traversing nodes until encountering the accessed node.
Accessing the node at position $i$ in the list costs $i$ (accessing the first node in the list costs $1$).

\vspace{-0.5cm}
\paragraph*{The precedence constraints (dependencies).}
We are given a directed acyclic graph \DAG, ofter called the \emph{dependency graph}.
The dependency graph induces a~partial order among the nodes that is equivalent to the reachability relation in \DAG.
The nodes must obey the partial order in the list at all time.
We say that a node $v$ is a~\emph{dependency} of a~node $u$ if there exists an edge $(u,v)$ in \DAG. Then, in every configuration of the list, $v$ must be in front of $u$.
We assume that the given initial configuration of the nodes obeys the precedence constraints induced by \DAG.

\vspace{-0.5cm}
\paragraph*{Node rearrangement.}
After serving a request, an algorithm may choose to rearrange the nodes of the list.
Precisely, the algorithm may perform any number of \emph{feasible} transpositions of neighboring nodes, i.e., transpositions that respect the precedence constraints induced by~\DAG.
We study the \emph{paid exchange model} where all transpositions incur the cost 1; this is different from the
\emph{free exchange model} sometimes considered in the literature where moving the requested node closer to the front is free.

\medskip

Our goal is to design online algorithms that perform closely to offline optimal algorithms.
For a more detailed description of the model in both the deterministic and randomized setting, see Appendix~\ref{sec:online}.

\subsection{Our Contributions}
\label{sec:contrib}

We initiate the study of a natural and practically motivated generalization of the online list access problem
where there can be precedence constraints among nodes.
Our main contribution are two constant-competitive online algorithms for this problem.
Our algorithms are designed around a recursive procedure Move-Recursively-Forward that generalizes the Move-To-Front
algorithm, accounting for dependencies.  We also shed light on the challenges of supporting insertions and deletions in the setting with precedence constraints.

\subsection{Novelty and Related Work}

Already various online problems have been studied in settings with dependencies and precedence constraints.
In scheduling with precedence constraints~\cite{Azar2002}, a job may be scheduled only after all its predecessors are completed. Another example which is more closely related to our work is caching with dependencies~\cite{Bienkowski2017}: the problem is motivated by the fact that routers in communication networks forward packets based on the longest common prefix match, and hence an element can be brought into the cache only if all its dependencies are in the cache.
However, we are not aware of any work on online list access problems with dependencies.

List access problems were already studied for several different cost models. Much prior work considers the \emph{paid exchange model} (studied in this paper) in which every transposition incurs a cost to the algorithm. For the setting without dependencies, it is known that no deterministic algorithm can be better than $3$-competitive; this lower bound is due to Reingold~et~al.~\cite{Reingold1994}.
The survey~\cite{KamaliSurvey2013} suggests that the deterministic algorithm Move-To-Front-every-other-access can be shown to be
$3$-competitive. In the randomized setting, the best known algorithm is an extension of COUNTER algorithm~\cite{Reingold1994, albers1997revisiting} that is $\frac{29}{11} \approx 2.63$-competitive against oblivious adversaries, and a lower bound (in the paid exchange model) against the oblivious adversary is 1.8654~\cite{Albers2020}.

There exist interesting results on list access problems with lower list rearrangement cost.
The most popular variant is the \emph{free exchange model}, in which moving an accessed node forward is free.
In such a setting, the algorithm Move-To-Front is $2$-competitive, and this result is tight~\cite{sleator1985amortized}.
Other papers considered settings in which rearrangements of large portions of the list have linear costs~\cite{Munro2000,Kamali2013}.
The list access problem was also studied under a \emph{generalized access cost} model~\cite{sleator1985amortized}, where the cost of accessing an $i$-th node is a general function $f(i)$.

\medskip

The main technical challenge and novelty in our paper is the design of an efficient operation for moving nodes closer to the front of the list.
The operation Move-To-Front, known from the setting without dependencies, would violate the order of nodes, and new algorithmic ideas are needed that account for dependencies.
Our candidate procedure Move-Recursively-Forward (cf.~Section~\ref{sec:det4competitive}) moves the accessed node forward, but additionally it moves forward a carefully chosen set of its dependencies, while retaining the linear cost of node rearrangements.

The candidate operation is expected to decrease the cost of future accesses.
To capture this property, we measure the distance to an optimal offline algorithm, and count the number of \emph{inversions} of pairs of nodes.
We argue that some crucial types of inversions are destroyed during the operation, to show that the algorithm is making progress towards the optimal solution.
The challenge lies in doing so without knowing the optimal solution's configuration.

\subsection{A Practical Motivation}

List access with precedence constraints is motivated by practical applications in the context of packet
classification~\cite{gupta2001algorithms},
a fundamental task performed by switches, routers and middleboxes (e.g., firewalls) in communication networks~\cite{gupta2001algorithms,Srinivasan1999,Eppstein2001,Hamed2006,Acharya2007}:
upon the arrival of a packet, its header it inspected in order to determine to which flow it belongs, and hence which predefined rule needs to be applied to process it.

A common application for packet classifiers is implementing traffic filtering in \emph{network firewalls}, where the rules distinguish legitimate packets, which are to be accepted by the firewall, from malicious traffic that needs to be dropped \cite{Hamed2006,Acharya2007}.

These rules often have dependencies, e.g., the matched domains overlap and the rules need to be performed in a specific (partial) order, i.e., a rule assumes that another rule has been checked before.
Figure~\ref{fig:large-dag} depicts an example of a dependency DAG induced by a table of rules.

\begin{figure}[h]
    \center
    \includegraphics[width=1.0\textwidth]{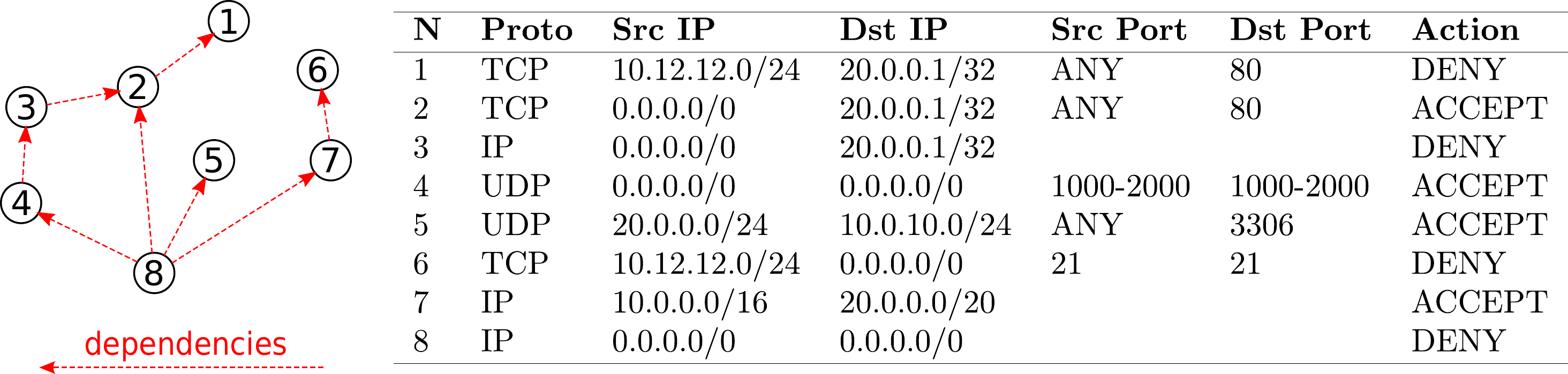}
    \caption{An example of a table of packet classification rules (right) and the corresponding dependency DAG \DAG among the nodes of the list (left).}
    \label{fig:large-dag}
\end{figure}

Packet classification is implemented using processing pipelines where match-action rules are organized in a data structure such as a linked list. A packet is first matched against the rule at the head of the list, and then, depending on the action, passed on to be classified further by the rules down in the list, according to the list order (respecting the precedence constraints).
With this perspective, matching the packet to its rule in a list can be seen as the find operation for a node in a list.

Due to its simplicity, this linear lookup structure is commonly applied in practice; e.g., in the default firewall suite of the Linux operating system kernel called \texttt{iptables} \cite{10.1145/3371927.3371929}, the OpenFlow reference switch \cite{openflow}, and in many QoS classifiers.
This approach however can be inefficient if a significant part of the list needs to be traversed before the final packet classification can be performed.

The motivation behind our optimization model is to render the data structure used for packet classification self-adjusting: if frequently used rules appear closer to the head of the list, the overhead of list traversal could be improved significantly, and hence packet classification sped up. This requires monitoring the importance of rules, which may change over time, and dynamically promote rules accordingly, while accounting for the precedence constraints to ensure policy compliance.

\section{Deterministic Algorithm}
\label{sec:det4competitive}

In this section we propose a deterministic 4-competitive algorithm \DET for online list access with dependencies.
We design the algorithm so that the cost of reorganizing the list after access (by operation \emph{Move-Recursively-Forward}) is in the order of the cost of the access (Lemma~\ref{lem:rearrangement-linear}).
To use the potential function analysis framework, introduced by Sleator and Tarjan~\cite{sleator1985amortized}, we assure that the algorithm properly influences the potential change to bound the amortized cost of the algorithm (we elaborate in Section~\ref{sec:bounding-inversions}).

\vspace{-0.2cm}
\paragraph*{Algorithm \DET.}
The algorithm uses a recursive procedure \emph{Move-Recursively-Forward} (\MRF for short).
The procedure \MRFx{$y$} moves the node $y$ forward (by transposing it with the preceding nodes) until it encounters any of its dependency nodes, say $z$, and recursively calls \MRFx{$z$}.
Upon receiving an access request to a node $\sigma_t$, \DET locates $\sigma_t$ on the list and invokes the procedure \MRFx{$\sigma_t$}.
In Figure~\ref{fig:follow}, we depict an example run of \MRF after serving a request by \DET.

\begin{figure}[h]
    \center
    \includegraphics[width=0.80\textwidth]{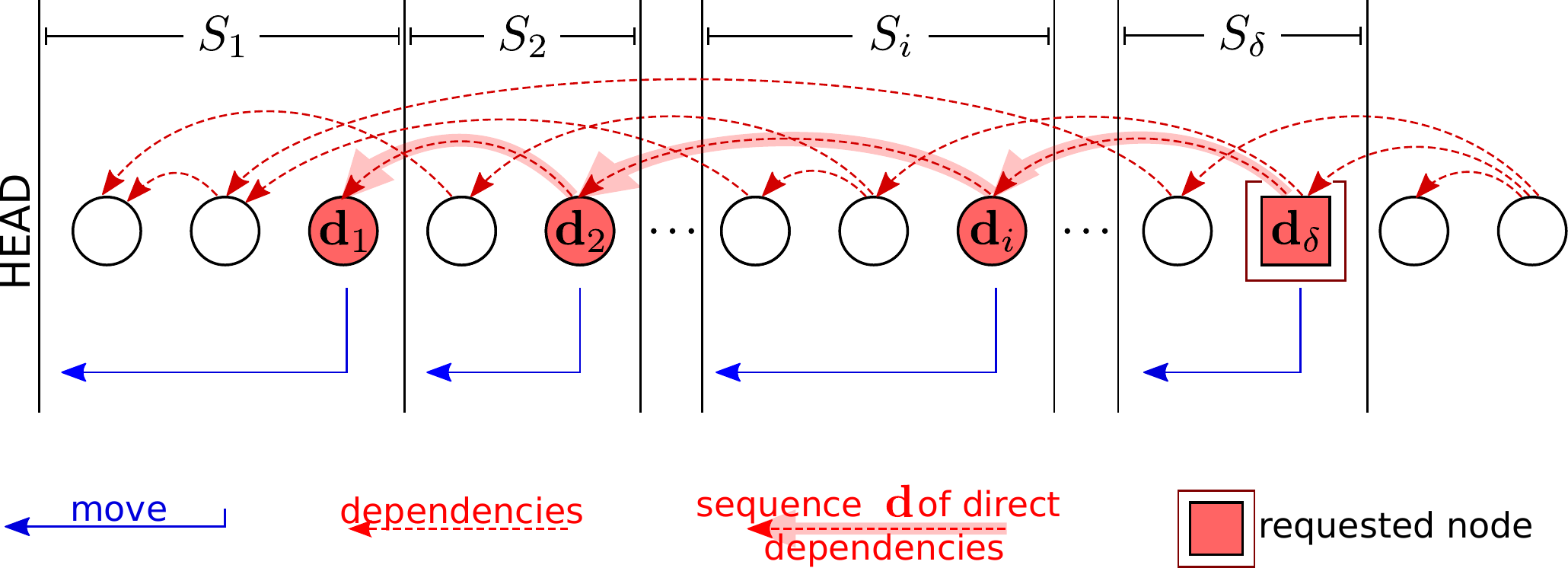}
    \caption{An example of handling a request by the algorithm \DET. The red nodes $\dep_i$, formally introduced in Section~\ref{sec:det-overview}, are the nodes that the algorithm moves forward (to the position denoted by a blue arrow). Among multiple dependencies (red dashed arrows), the algorithm \DET moves forward only the nodes that constitute the path of the direct dependencies (bold red arrow).}
    \label{fig:follow}
\end{figure}

\medskip
We present the pseudocode of \DET in Algorithm~\ref{alg:follow}.
We say that a node $y$ is a~\emph{direct dependency} of a node $z$ if $y$ is the dependency of $z$ that is located at the furthest position on the list.
By $\pos(z)$ we denote the position of node $z$ in the list maintained by the algorithm, counting from the head of the list
(recall that the position of the first node is $1$).

\medskip
\begin{algorithm}[H]
 \caption{The algorithm \DET.}
 \label{alg:follow}
\SetAlgoNoLine
\SetKwInOut{Input}{Input}

\Input{An access request to node $\sigma_t$}

  Access $\sigma_t$

  Run the procedure \MRFx{$\sigma_t$}\\

  \BlankLine
   \BlankLine

\SetKwInput{Procedure}{Procedure \MRFx{$y$}}
\Procedure{}
\Indp
 	\uIf{$\text{\upshape  $y$ has no dependencies}$}{
		Move $y$ to the front of the list
	}
 	\uElse{
		Let $z$ be the direct dependency of $y$\\
		Move node $y$ to $\pos(z)+1$\\
		Run the procedure \MRFx{$z$}
  }
\end{algorithm}

\medskip

We claim that this algorithm is $4$-competitive.
Before proving the claim (Theorem~\ref{thm:follow-4-competitive}), we overview the analysis, introduce sets relevant to the analysis and reason about their size.

\subsection{Analysis overview}
\label{sec:det-overview}
To perform an amortized analysis of \DET's cost, we use a potential function similar to the one used in the analysis of Move-To-Front~\cite{sleator1985amortized}, defined in terms of number of inversions.
We~bound the change in the potential due to node rearrangements by a similar function as in Move-To-Front's analysis (Theorem~\ref{thm:change-inversions}).
In our case multiple nodes change their positions, and the proof requires a careful analysis of inversions affected by each of them.
We note that the cost of node rearrangements after accessing a node is bounded by the cost of accessing the requested node (Lemma~\ref{lem:rearrangement-linear}).

Before analyzing the competitive ratio of the algorithm, we introduce the notation and the sets and sequences of nodes relevant to our analysis.
\vspace{-0.3cm}
\paragraph*{Inversions.}
An \emph{inversion} is an ordered pair of nodes $(u, v)$ such that $u$ is located before $v$ in \DET's list and $u$ is located after $v$ in \OPT's list.
The inversion is the central concept in the analysis of the presented algorithms in this paper.
\vspace{-0.3cm}
\paragraph*{The rearranged nodes $\dep_j$.}
Consider a single request to a node $\sigma_t$ and the node rearrangements at $t$.
Let~$\dep$ be the sequence of the nodes that the algorithm moves forward (calls \MRF for), ordered by increasing distance to the head.
Let $\depx$ be the length of~$\dep$.
We~emphasize that $\dep$ contains the requested node at the last position, $\sigma_t = \dep_\depx$.
\vspace{-0.3cm}
\paragraph*{Values $k$ and $\ell$.}
To compare the cost of \DET and \OPT, we define values $k$ and $\ell$ related the number of nodes in front of the requested node $\sigma_t$ in \DET's and \OPT's list. Precisely,
 let $k$
be the number of nodes before $\sigma_t$ in both \DET's and \OPT's lists,
and let $\ell$ be the number of nodes before $\sigma_t$ in \DET's list, but after $\sigma_t$ in \OPT's list.
\vspace{-0.3cm}
\paragraph*{Sets $K_j$ and $L_j$.}
With the values $k$ and $\ell$ it is possible to analyze the classic algorithm Move-To-Front, yet they are not sufficient to express the complexity of Move-Recursively-Forward.
Hence, we generalize the notion of $k$ and $\ell$ to sets of elements related to positions of individual nodes $\dep_j$ in \DET's and \OPT's lists.
Precisely, let $K_j$ be the set of elements before $\dep_j$ in both \DET's and \OPT's lists for $j \in [1, \depx]$,
    and let $L_j$ be the set of elements before $\dep_j$ in \DET's list but after $\dep_j$ in \OPT's list.
We note that these sets are generalizations of $k$ and $\ell$: for the accessed node $\dep_\depx$ we have $k = |K_\depx|$ and $\ell = |L_\depx|$.
\vspace{-0.3cm}
\paragraph*{Sets $S_j$.}
The sets of nodes between the nodes $\dep$ in \DET's list are crucial to the analysis.
Intuitively, the node $\dep_i$ moves in front of all the nodes from the set $S_i$.
Let $S_1$ be the elements between the head of \DET's list and $\dep_1$ (included).
For $j \in [2, \depx]$, let $S_j$ be the set of elements between $\dep_j$ and $\dep_{j-1}$ (with $\dep_{j-1}$ excluded) in \DET's list.

\medskip

Figure~\ref{fig:sets} illustrates an example of  possible composition of sets $K_j$, $L_j$ and $S_j$ for different values of $j$ on a given access request.
\begin{figure}[h]
    \label{fig:sets}
    \center
    \includegraphics[width=0.55\textwidth]{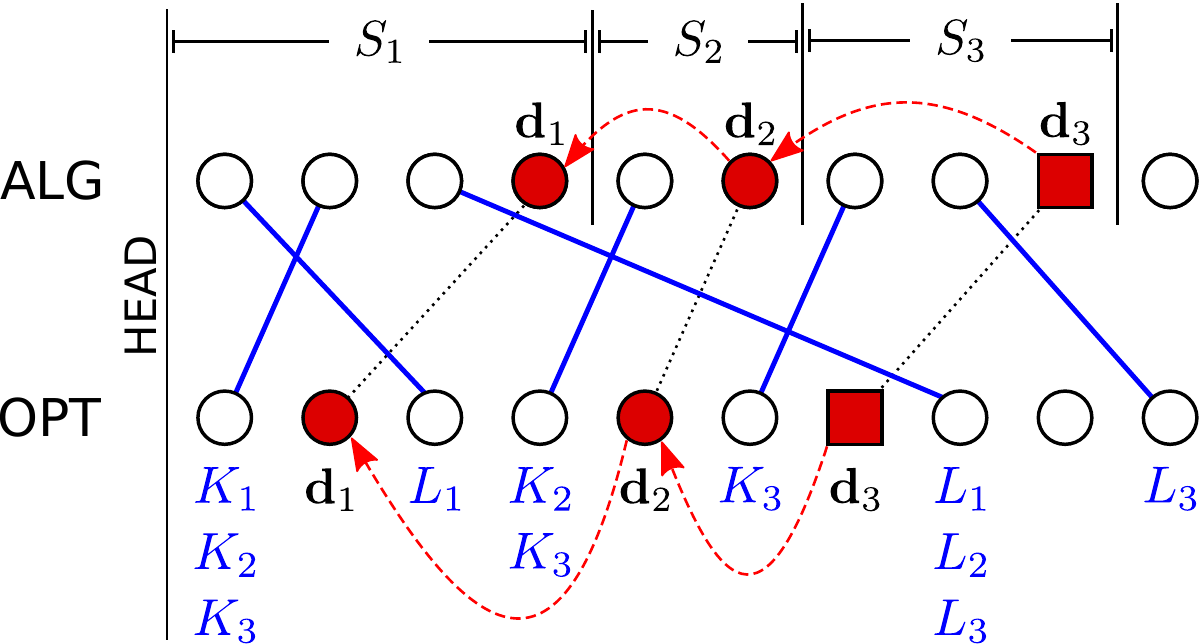}
    \caption{The example illustrates central definitions of sets of nodes used in our analysis. We~depict the positions of nodes in both \DET's and \OPT's list (joined by solid blue lines). The dotted black lines between the nodes $\dep_i$ help in determining the assignment of nodes to sets: in $K_i$ we have the nodes in front of the dotted black line between $\dep_i$, and in $L_i$ we roughly have the nodes that cross the dotted black lines between $\dep_i$'s.}
    \label{fig:sets}
\end{figure}

\subsection{Bounding the Change of Inversions}
\label{sec:bounding-inversions}

The operation Move-Recursively-Forward, defined as a procedure in Algorithm~\ref{alg:follow} was designed to mimic the operation Move-To-Front~\cite{sleator1985amortized} in introducing and destroying inversions.

\begin{theorem}
    \label{thm:change-inversions}
    Consider a request to the node $\sigma_t$, and fix a configuration of \OPT at time~$t$.
    Then, the change in the number of inversions due to \DET's node rearrangement after serving the request is at~most~$k-\ell$.
\end{theorem}

To prove this claim, we consider the influence of the Move-Recursively-Forward operation on values $k$ and~$\ell$ (defined for the currently requested node) by inspecting the sets $K_j$ and $L_j$ (defined for the nodes $\dep_j$).
We separately bound the number of inversions created (Lemma~\ref{lem:created-inversions}) and destroyed (Lemma~\ref{lem:destroyed-inversions}).
Before showing these claims, we inspect the basic relations between the sets $K_j$, $L_j$ and $S_j$ (Lemma~\ref{lem:set-relations}).
\begin{restatable}{Lemma}{setRelationsLemma}
    \label{lem:set-relations}
    The following relations hold
    \begin{enumerate}
        \item \label{it:ki}
        $
        \bigcup_{j=1}^{\depx} K_j  =  K_{\depx},
        $
        \item \label{it:li}
        $
        \bigcup_{j=1}^{\depx}  (S_j \cap L_j)  =  \bigcup_{j=1}^{\depx} L_j .
        $
    \end{enumerate}
\end{restatable}
See Figure~\ref{fig:sets} for an illustration of a graphical argument. 

\begin{proof}
    First, we prove the equality \ref{it:ki}.
We show inclusions both ways.
Note that the order between nodes from $\dep$ is the same in both \DET's and \OPT's lists.
Hence, a node $y \in K_i$ is in front of $\dep_i$ and in front of all $\dep_j$ for $i \le j$ in both \DET's and \OPT's list.
Consequently, each~node from $K_i$ belongs to all $K_j$ for $i \le j$, and we have $\bigcup_{j=1}^{\depx} K_j \subseteq K_{\depx}$.
Conversely, $K_{\depx} \subseteq \bigcup_{j=1}^{\depx} K_j$ by basic properties of sets, and we conclude that the sets are equal, and the equality holds.

\medskip

    Next, we prove the equality \ref{it:li}.
    We show inclusion both ways.
    Consider any element $y \in L_j$.
    The sets $\{S_j\}$ partition the nodes placed closer to the front of the list than $\sigma_t$ (the requested node), thus $y$ belongs to some $S_i$ for $i \le j$.
    Fix such $i$; we claim that additionally $y \in L_i$:
    \begin{itemize}
        \item $y$ belongs to $S_i$, and hence it is in front of $\dep_i$ in \DET's list,
        \item $y$ is after $\dep_j$ in \OPT's list (it belongs to $L_j$), and hence it is after $\dep_i$ in \OPT's list (the order of $\dep$ is fixed due to dependencies).
    \end{itemize}
    Hence, any $y \in L_j$ belongs to $S_i \cap L_i$ for some $i$, and we conclude that the inclusion $\bigcup_{j=1}^{\depx} L_j \subseteq \bigcup_{j=1}^{\depx}  (S_j \cap L_j)$ holds.
    Conversely, by properties of sets $\bigcup_{j=1}^{\depx}  (S_j \cap L_j) \subseteq \bigcup_{j=1}^{\depx} L_j$, and we conclude that the sets are equal and the equality holds.
\end{proof}

\begin{lemma}
    \label{lem:created-inversions}
    Consider a request to the node $\sigma_t$, and fix a~configuration of \OPT at time $t$.
    Due to rearrangements after serving the request, \DET creates at most $k$ inversions.
\end{lemma}

\begin{proof}
    Let $\Ic_j$ be the number of inversions added by moving a single node $\dep_j$ by \DET, for $j \in [1, \depx]$.
    To bound $\Ic_j$, we inspect the set $S_j$ of the nodes that $\dep_j$ overtakes, and we reason based on their positions in \OPT's list.
    Moving $\dep_j$ forward creates inversions with nodes in (possibly a subset of) $S_j \cap K_j$.
    No other node changes its relation to the set $S_j$, hence the inversions for nodes in $S_j$ are influenced only by the movement of $\dep_j$.
    This gives us the bound
    $ \Ic_j \leq  | S_j \cap K_j |$.

    We sum up the individual bounds on $\Ic_j$ for all $j$ to bound the total number of inversions created%
$$ \sum_{j=1}^{\depx} \Ic_j \leq  \sum_{j=1}^{\depx} | S_j \cap K_j |  = \lvert \bigcup_{j=1}^{\depx}  (S_j \cap K_j) \rvert \le | \bigcup_{j=1}^{\depx} K_j |,$$
where the second step holds as the sets $S_j$ are disjoint, and the last step follows by basic properties of sets.
By Lemma~\ref{lem:set-relations}, equation \ref{it:ki}, we have
$|\bigcup_{j=1}^{\depx} K_j | =  |K_{\depx}|$, thus by combining the above inequalities, we have $ \sum_{j=1}^{\depx} \Ic_j \leq  |K_{\depx}| = k$, and we conclude that the claim holds.
\end{proof}

\begin{lemma}
    \label{lem:destroyed-inversions}
    Consider a request to the node $\sigma_t$, and fix a~configuration of \OPT at time $t$.
    Due to rearrangements after serving the request, \DET destroys at least $\ell$ inversions.
\end{lemma}

\begin{proof}
    Let $\Id_j$ be the number of inversions destroyed by moving a single node $\dep_j$ by \DET, for $j \in [1, \depx]$.
    To bound $\Id_j$, we inspect the set $S_j$ of the nodes that $\dep_j$ overtakes, and we reason based on their positions in \OPT's list.
    Moving $\dep_j$ forward destroys all inversions with nodes in  $S_j \cap L_j$.
    No other node changes its relation to the set $S_j$, hence the inversions for nodes in $S_j$ are influenced only by the movement of $\dep_j$.
    This gives us the bound
    $ \Id_j \ge | S_j \cap L_j |$.

    We sum up the individual bounds on $\Id_j$ for all $j$ to bound the total number of inversions destroyed.
    $$ \sum_{j=1}^{\depx} \Id_j \geq   \sum_{j=1}^{\depx}  | S_j \cap L_j |  =  \lvert \bigcup_{j=1}^{\depx}  (S_j \cap L_j) \rvert,$$
    where the second step holds as the sets $S_j$ are disjoint.
    By Lemma~\ref{lem:set-relations}, equation~\ref{it:li} we have
    $|\bigcup_{j=1}^{\depx}  (S_j \cap L_j)|  =  |\bigcup_{j=1}^{\depx} L_j |$.
    Finally, by basic properties of sets, $ | \bigcup_{j=1}^{\depx} L_j | \ge |L_{\depx}| = \ell$,
    and by combining all the above bounds we have $\sum_{j=1}^{\depx} \Id_j \geq   \ell$,
    and we conclude that the lemma holds.
\end{proof}

Combining Lemmas~\ref{lem:created-inversions} and \ref{lem:destroyed-inversions} gives us the joint bound on the change in the number of inversions, and proves the Theorem~\ref{thm:change-inversions}.
We note that this bound is consistent with the bound on the changes in inversions for the algorithm Move-To-Front~\cite{sleator1985amortized}, where the inversions were considered with respect to the accessed node only.

\subsection{Bounding the Competitive Ratio}
\label{ssec:ratio}

Finally, we show the main result of this section: the competitive ratio of the algorithm \DET is 4.
First, we observe that the cost of the rearrangements after handling each request is bounded by the cost of the access (Lemma~\ref{lem:rearrangement-linear}), and then we apply the bounds on the number of inversions (Theorem~\ref{thm:change-inversions}), to finally bound the ratio using a potential function argument (Theorem~\ref{thm:follow-4-competitive}).
\begin{restatable}{Lemma}{rearrangementLinearLemma}
    \label{lem:rearrangement-linear}
    Consider a single request to a node $y$ at position $\pos(y)$ handled by the algorithm \DET.
    The rearrangements after serving the request to $y$ costs at most $\pos(y)$.
\end{restatable}
Intuitively, each of the nodes that was moved forward goes through a disjoint part of the list, in total at most $\pos(y)$.
For~a~graphical argument, see Figure~\ref{fig:follow}.

\begin{proof}
    Recall that $y$ is moved to the position right before its furthest dependency, and then recursively the dependency moves forward until encountering the dependency of its own.
    The movements end when a moving node reaches the front of the list.
    Each node is moved right to a position \emph{one place behind} its dependency.

    Let $\dep$ be the set of the dependencies of $y$, and let $\depx$ be the number of $y$'s dependencies (including $y$).
    Then, the node $\dep_i$ moves to the position $\pos(\dep_{i-1})+1$ (for formality of the argument, we assume an artificial dependency at the head of the list, $\pos(\dep_0) = 0$).
    Thus, in total the number of transpositions is
    \[
        \sum_i^\depx (\pos(\dep_i) - (\pos(\dep_{i-1})+1)) = \pos(\dep_\depx) - \pos(\dep_0) - \depx \le \pos(\dep_\depx).
    \]
    As $y = \dep_\depx$, we conclude that the lemma holds.
\end{proof}

\vspace{-0.3cm}
\paragraph*{Events overview.}
In the analysis, we distinguish between the following types of events that occur throughout algorithms' execution:
\begin{enumerate}[label=(\Alph*)]
    \item
    A \emph{request event} $R^i(\sigma_t)$ for $i \in \{0, 1\}$. The algorithm serves the request to the node~$\sigma_t$ and runs the Move-Recursively-Forward procedure.
    We assume a fixed configuration of \OPT throughout this event.
    \item A \emph{paid exchange event} of OPT, $P(\sigma_t)$, a single paid transposition performed by \OPT, where it either creates or destroys a single inversion with respect to the node~$\sigma_t$.
    We assume a fixed configuration of \DET throughout this event.
\end{enumerate}
\vspace{-0.7cm}
\paragraph*{Potential function.}
We define the potential function $\Phi$ in terms of number of inversions in \DET's list with respect to OPT's list.
Precisely, the potential function is defined as \emph{twice the number of inversions}.

Finally, we prove the main result of this section, that the algorithm \DET is $4$-competitive.

\begin{restatable}{Theorem}{followCompetitive}
    \label{thm:follow-4-competitive}
    The algorithm \DET is strictly 4-competitive.
\end{restatable}

The proof uses a potential function argument similar to Move-To-Front~\cite{sleator1985amortized}, and it internally uses the Theorem~\ref{thm:change-inversions} to reason about changes in inversion due to Move-Recursively-Forward runs.

\begin{proof}
    Fix a sequence of requests $\sigma$.
    We compare the costs of \DET and an optimal offline algorithm \OPT on $\sigma$ using a potential function $\Phi$.
    Let $C_{\DET}(t)$ and $C_{\OPT}(t)$ denote the cost incurred at time $t$ by \DET and \OPT respectively.

    First, we bound the cost of \DET incurred while serving an access request to a node $\sigma_t$ at time $t$ (a \emph{request event}).
    This cost consists of the access cost and the rearrangement cost.
    To access the node $\sigma_t$, the algorithm incurs the cost $\pos(\sigma_t)$, and by Lemma~\ref{lem:rearrangement-linear} the rearrangement cost is bounded by $\pos(\sigma_t)$, hence
    $C_{DET}(t) \le 2\cdot \pos(\sigma_t)$.

    Next, we bound the amortized cost for every access request served by \DET.
    The amortized cost is $C_{\DET}(t) + \Delta \Phi(t)$ for each time $t$.
    By Theorem~\ref{thm:change-inversions}, we bound the change in the number of inversions due to \DET's rearrangement after serving the request at time $t$ by $\Delta I \le k - \ell$.
    Thus, the change in the potential is $\Delta \Phi(t) \le 2(k - \ell)$.
    As~$\pos(\sigma_t) = k + \ell + 1$, combining these bounds gives us
    \begin{equation*}
    C_{\DET}(t) + \Delta \Phi(t) \le  2\cdot \pos(\sigma_t) + 2 \left(  k - \ell \right)
        \le 2 \left(k+\ell+1\right)+ 2 \left(k - \ell \right)
        \le 4 \cdot C_{\OPT}(t),
    \end{equation*}
    where the last inequality follows by $C_{\OPT} \ge k + 1$.

    Note that the bound on amortized cost accounts for possible \emph{paid exchange events}, the rearrangement of \OPT at time $t$.
    Each transposition of \OPT increases the number of inversions by $1$, which increases the LHS by $2$;
    and for each transposition \OPT pays $1$, which increases the RHS by $4$.

   Finally, we sum up the amortized bounds for all requests of the sequence $\sigma$ of length $m$, obtaining
    \begin{equation*}
    C_{\DET}(\sigma) + \Phi(m) - \Phi(0)   \le 4\cdot C_{\OPT}(\sigma).
    \end{equation*}
    We assume that \DET and \OPT started with the same list, thus the initial potential $\Phi(0) = 0$, and the potential is always non-negative, thus in particular $\Phi(m) \ge 0$, and we conclude that
    $C_{\DET}(\sigma) \le 4\cdot C_{\OPT}(\sigma)$.\qedhere
\end{proof}

\section{Randomized Algorithm}
\label{sec:rfollow}

In this section we propose a $3$-competitive randomized algorithm \RAND for online list access with dependencies.
We design the algorithm around concepts from the algorithm~BIT~\cite{Reingold1994} and the procedure Recursively-Move-Forward from Section~\ref{sec:det4competitive}.

\pagebreak
\medskip
\noindent
\textbf{Algorithm \RAND.}
We maintain an additional bit of memory for each node: we assign a~binary counter $\bit(y)$ to each node $y$ with initial value chosen uniformly at random from $\{0, 1\}$.
Upon receiving a~request to a node $\sigma_t$, if its bit was $0$, we call the procedure \MRFx{$\sigma_t$}, defined in Algorithm~\ref{alg:follow}, and then we flip the bit of $\sigma_t$ regardless of its previous value.
We~present the pseudocode of \RAND in Algorithm~\ref{alg:r-follow}.

\medskip
\begin{algorithm}[H]
 \caption{The randomized algorithm \RAND.}
  \label{alg:r-follow}
\SetAlgoNoLine
 \SetKwInput{Initialization}{Initialization~}
 \Initialization{Assign a bit to each node uniformly at random }
\SetKwInput{Input}{Input~}
\Input{~A request to node $\sigma_t$}
Access $\sigma_t$\\
	  \If{$\text{\upshape  $\bit(\sigma_t)$ is 0}$}{
	   	Run the procedure \MRFx{$\sigma_t$}
	   }
Flip $\bit(\sigma_t)$\\

\end{algorithm}

\vspace{-0.1cm}
\paragraph*{Analysis.}
We compare the costs of \RALG and an optimal offline algorithm \OPT on $\sigma$ using the potential function $\Phi$.
In the definition of $\Phi$, we distinguish between two types of inversions; we say that an inversion of the ordered pair $(y,\sigma_t)$ is a \emph{type $i$ inversion} if $\bit(\sigma_t)=i$ for $i\in\{0,1\}$.
Precisely, we use the potential function
\begin{equation*}
    \Phi= \frac{5}{2}\cdot I_0 + \frac{7}{2}\cdot I_1,
\end{equation*}
where $I_0$ is the number of inversions of type $0$ and $I_1$ is the number of inversions of type $1$.

In proving our claims, we use an observation that the values bits of \RAND are independent and remain uniformly distributed as they change over time.

\begin{observation}
    \label{obs:independent-bits}
    	For any node $y$, the value of $\bit(y)$ is $0$ with probability $\frac{1}{2}$ and $1$ with probability $\frac{1}{2}$ at any time, and is independent of its position in \OPT's list and other nodes' bits.
\end{observation}

We distinguish between the following types of events that occur throughout algorithm's execution:
\begin{enumerate}[label=(\Alph*)]
    \item
    An \emph{access request event} $R^i(\sigma_t)$ for $i \in \{0, 1\}$. The algorithm serves the request to the node~$\sigma_t$ with $\bit(\sigma_t)=i$ and performs the Move-Recursively-Forward procedure if $\bit(\sigma_t)=0$.
    We assume a~fixed configuration of \OPT throughout this event.
    \item A \emph{paid exchange event} of OPT, $P(\sigma_t)$, a single paid transposition performed by \OPT, where it either creates or destroys a single inversion with respect to the node~$\sigma_t$.
    We assume a fixed configuration of \RALG throughout this event.
\end{enumerate}

\medskip

We start by analyzing the amortized cost of \RALG for a single access request event (Lemma~\ref{lem:rfollow-amortized}).
If the accessed node's bit has value 1, then none of the nodes change their position, but some inversions may change their type and entail the change in potential, accounted in the amortized cost.
If the accessed node's bit has value 0, then the algorithm runs the procedure Move-Recursively-Forward, and we bound the amortized cost of this rearrangement.

Let $C_{\RALG}(t)$ and $C_{\OPT}(t)$ be the cost incurred at time $t$ by \RALG and \OPT, respectively.
Recall that $k$
is the number of nodes before $\sigma_t$ in both \RALG's and \OPT's lists,
and $\ell$ is the number of nodes before $\sigma_t$ in \RALG's list, but after $\sigma_t$ in \OPT's list.
For succinctness, we define $E_0[Y]$ as the expected value of a~variable $Y$ when $\bit(\sigma_t)=0$, and $E_1[Y]$ as the expected value of $Y$ when $\bit(\sigma_t)=1$.
\pagebreak
\begin{restatable}{Lemma}{rfollowAmortized}
    \label{lem:rfollow-amortized}
    Consider a request to a node $\sigma_t$ handled by the algorithm \RAND.
    \begin{enumerate}
        \item If $\bit(\sigma_t)=1$, then $E_1[C_{\RALG}(t)+\Delta\Phi]  \le k+1$.
    \item If $\bit(\sigma_t)=0$, then $E_0[C_{\RALG}(t)+\Delta\Phi]\le 5(k+1)$.
    \end{enumerate}
\end{restatable}

\begin{proof}
    Before bounding the amortized cost, we split the change in potential $\Delta\Phi$,
    distinguishing between three reasons for the potential to change: (1) the inversions added, (2) the inversions destroyed and (3) the inversions that changed their type.
    Precisely, we split the change of the potential into three parts, $\Delta \Phi = A+D+F$, where $A$ is a~random variable denoting the change in potential due to new inversions added, $D$ is a~random variable denoting the change in potential due to old inversions destroyed and $F$ is a random variable denoting the change in potential due to inversions that flipped (changed) their type.
    Together, these three parts account for all the changes while serving a request by \RALG.

    \medskip
    Recall that $E_0[Y]$ is the expected value of a~variable $Y$ when $\bit(\sigma_t)=0$, and $E_1[Y]$ is the expected value of $Y$ when $\bit(\sigma_t)=1$.

    Consider the case $\bit(\sigma_t)=1$.
    The algorithm flips the bit of $\sigma_t$ to 0, and it does not perform any node rearrangements.
    No inversions are destroyed or added, thus
    $E_1[D] = 0$ and
    $E_1[A] = 0$.
    The algorithm incurs the cost $k + \ell + 1$ due to access, and type of $\ell$ inversions flipped from $1$ to $0$, therefore
    \[
    E_1[C_{\RALG}(t)+A+D+F]  \le (k + \ell + 1) - \ell=k+1,
    \]
    which concludes the first claim of the lemma.

    \medskip

    Consider the case $\bit(\sigma_t)=0$.
    The algorithm invokes the \MRF procedure and then flips the bit of $\sigma_t$ to $1$.
    We split the amortized cost into four parts:
    \begin{enumerate}
        \item \textbf{The cost incurred by the algorithm.}

    The algorithm incurs the cost $k+\ell+1$ for access and by Lemma~\ref{lem:rearrangement-linear} at most $k+\ell+1$ for rearrangements, thus
    \[
        E_0[C_{\RALG}(t)] \le 2\cdot (k+\ell+1).
    \]
    \item
    \textbf{The change in potential due to inversions that flipped type from 0 to 1.}
    Out of all $\ell$ inversions with respect to $\sigma_t$, some are destroyed and the others flipped type.
    The bit of $\sigma_t$ changes to $1$, thus some inversions may have flipped type from 0 to 1.
    Precisely, all inversions $(y,\sigma_t)$ for $y\in L_{\delta}$ flip their type unless $\sigma_t$ moves in front of $y$.
    Any single such flip from type 0 to type 1 increases the potential by $\frac{7}{2} - \frac{5}{2} = 1$.
    The node $\sigma_t$ moves in front of the set of nodes $S_{\delta}$, thus
    \[
        E_0[F] \le (\frac{7}{2}-\frac{5}{2})\cdot(\ell-|S_{\delta}\cap L_{\delta}|) = \ell - |S_{\delta}\cap L_{\delta}|.
    \]

    \pagebreak
    \item
    \textbf{The change in potential due to inversions destroyed.}

    When the procedure Move-Recursively-Forward is invoked, each node $\dep_i$ for $i=\{1,2,\ldots, \delta\}$ moves in front of the nodes from $S_i$ (recall that $\sigma_t = \dep_{\delta}$).
    Thus, the inversions $S_i \cap L_i$ are destroyed, and the change in the potential depends on the bit of $\dep_i$.
    The bit of $\sigma_t$ is determined, $b(\sigma_t)=0$, and by Observation~\ref{obs:independent-bits}, the bits of $\dep_i$ for $i < \delta$ are equally likely to be $0$ or $1$.
    If the bit of the moving node is $0$, a single inversion destroyed contributes $-\frac{5}{2}$ to the potential change, otherwise it contributes $-\frac{7}{2}$, and we bound the total change in the potential by
    \begin{align*}
        E_0[D]&\le \underbrace{-\frac{5}{2}\cdot(|S_{\delta}\cap L_{\delta}|)}_{\sigma_t=\dep_{\delta}} - \sum_{i=1}^{\delta-1}( \underbrace{\frac{1}{2}\cdot\frac{5}{2} + \frac{1}{2}\cdot\frac{7}{2}) \cdot |S_{i}\cap L_{i}|}_{\mbox{nodes }\dep_i\mbox{ for }i < \delta}\\
    &\le \quad -\frac{5}{2}\cdot(|S_{\delta}\cap L_{\delta}|) - 3\cdot \sum_{i=1}^{\delta-1}  |S_{i}\cap L_{i}|.
    \end{align*}

    \item
    \textbf{The change in potential due to inversions added.}

    When the procedure Move-Recursively-Forward is invoked, each node $\dep_i$ for $i=\{1,2,\ldots, \delta\}$ moves in front of the nodes from $S_i$ (recall that $\sigma_t = \dep_{\delta}$).
    Thus, the inversions $S_i \cap K_i$ will be added.
    By Observation~\ref{obs:independent-bits}, bits of the nodes
    are equally likely to be $0$ or~$1$.
    If bit of the node that caused inversion is $0$, a single inversion added contributes $\frac{5}{2}$ to the potential change, otherwise it contributes $\frac{7}{2}$, and we bound the total change in the potential in the following way

        \begin{align*}
        E_0[A] \le \sum_{i=1}^{\delta} \frac{1}{2}\cdot(\frac{5}{2}\cdot |S_{i}\cap K_{i}|) + \frac{1}{2}\cdot(\frac{7}{2} \cdot |S_{i}\cap K_{i}|) \le \sum_{j=1}^{k}\left(\frac{1}{2}\cdot\frac{5}{2} + \frac{1}{2}\cdot \frac{7}{2}\right) = 3 k,
    \end{align*}
    where in the second inequality we used
    $\sum_{j=1}^{\depx} | S_j \cap K_j |  = \lvert \bigcup_{j=1}^{\depx}  (S_j \cap K_j) \rvert \le | \bigcup_{j=1}^{\depx} K_j |$, since we know the sets $S_j$ are disjoint.
    \end{enumerate}

    Finally, we combine the bounds (1-4) for the case $b(\sigma_t)=0$.
    We highlight that it is necessary to bound the expected cost of $D$ and $F$ together, obtaining
    \begin{align*}
    E_0[ C_{\RALG}(t)+D+F] & \le \underbrace{2\cdot (k+\ell+1)}_{C_{\RALG}} + \underbrace{(\ell-|S_{\delta}\cap L_{\delta}|)}_{F} \underbrace{-\frac{5}{2}\cdot(|S_{\delta}\cap L_{\delta}|) - 3\cdot \sum_{j=1}^{\delta-1}  |S_{j}\cap L_{j}|}_{D}\\
    & \le 2(k+1) + 3\ell - 3\cdot(|S_{\delta}\cap L_{\delta}|) - 3\cdot \sum_{j=1}^{\delta-1}  |S_{j}\cap L_{j}|= 2 (k+1),
    \end{align*}
    where the last equality holds because of $ \bigcup_{j=1}^{\depx}  (S_j \cap L_j)  =  \bigcup_{j=1}^{\depx} L_j $ by Lemma~\ref{lem:set-relations} and $\ell = |\bigcup_{j=1}^{\delta}   L_{j}|$.
    Including the bound for added inversions for the case $\bit(\sigma_t)=0$ gives us the total expected cost $E_0[C_{\RALG}(t)+A+D+F]\le 5(k+1)$,
    which concludes the second claim of the lemma.
\end{proof}

\pagebreak
Now, we are ready to prove the main result of this section.

\begin{restatable}{Theorem}{rfollowCompetitive}
    \label{thm:rfollow}
    The algorithm \RAND is strictly $3$-competitive against an oblivious offline adversary.
\end{restatable}

\begin{proof}
    Consider a sequence of access requests $\sigma$.
    We bound the ratio of \RALG to \OPT using a potential function argument.

    First, we consider a \emph{request event}.
    We use the Lemma~\ref{lem:rfollow-amortized} to bound the amortized cost,
    and we combine the bounds for the cases $\bit(\sigma_t)=0$ and $\bit(\sigma_t)=1$.
    By Observation~\ref{obs:independent-bits}, the probability that $\bit(\sigma_t)=0$ is $\frac{1}{2}$, and the probability that $\bit(\sigma_t)=1$ is $\frac{1}{2}$.
    Hence, we bound the expected amortized cost of \RALG in the following way, and relate it to the cost of \OPT
    \begin{align*}
        E[C_{\RALG}(t)+\Delta\Phi(t)] & = \frac{1}{2} \cdot E_0[C_{\RALG}(t)+\Delta\Phi(t)] + \frac{1}{2} \cdot E_1[C_{\RALG}(t)+\Delta\Phi(t)]\\
        & \le \frac{1}{2}\cdot 5(k+1) + \frac{1}{2}\cdot (k+1)= 3 (k+1) \le 3\cdot C_{\OPT}(t),
    \end{align*}
    where the last inequality holds as $C_{\OPT}(t) \ge k+1$.

    Next, we show that this bound accounts for \emph{paid exchange events}, the paid rearrangements by  \OPT at time $t$ (the second type of event in our analysis).
    Each transposition of \OPT increases the number of inversions by $1$, which increases the LHS by $\frac{1}{2}\cdot \frac{5}{2}  + \frac{1}{2}\cdot \frac{7}{2} = 3$ in expectation, because bits of the nodes are independent of the  \OPT decisions, and both types of inversions are equally probable to be created.
    To complement the bound at the RHS, for each transposition \OPT pays $1$, which increases the RHS by $3$.

    Finally, we sum up the amortized bounds for all requests of the sequence $\sigma$ of length $m$, obtaining
    \begin{equation*}
        E[C_{\RALG}(\sigma) + \Phi(m) - \Phi(0)]   \le 3\cdot C_{\OPT}(\sigma).
    \end{equation*}
    We assume that \RALG and \OPT started with the same list, thus the initial potential $\Phi(0) = 0$, and the potential is always non-negative, therefore $\Phi(m) \ge 0$, and we conclude that
    $E[C_{\RALG}(\sigma)] \le 3\cdot C_{\OPT}(\sigma)$.\qedhere
\end{proof}

\medskip

Finally, we discuss the competitive ratio in comparison to well-known randomized algorithm for the classic list access.
The algorithm BIT (that inspired the design of \RAND) is 2.75-competitive in the paid exchange model without precedence constraints~\cite{Reingold1994}.
The algorithm \RAND achieves slightly worse competitive ratio than BIT due to inversions that change their type.
The algorithms differ in the operation on the accessed node: the operation Move-To-Front destroys all inversions with respect to the node that changes its bit, and the operation Move-Recursively-Forward may change the type of some inversions.

We note that the improvement beyond the competitive ratio 3 with randomized algorithms is non-obvious.
We inspected a wide range of known randomized algorithms (such as COUNTER, RANDOM RESET and Markov-Move-To-Front~\cite{garefalakis1997new,Reingold1994,albers1997revisiting}) and we report that simply changing the operation from Move-To-Front to Move-Recursively-Forward still results in the competitive ratio~3.
The best lower bound for list update in the paid exchange model against oblivious adversaries is 1.8654~\cite{Albers2020}.

\section{Handling Insertions and Deletions}
\label{sec:insertions}

In data structures, such as linked lists, the sets of nodes can change over time.
Many data structures (including the ones for packet classification) hence support not only access operations, but also insertions and deletions.

\vspace{-0.2cm}
\paragraph*{Model.}
Consider the online list access with dependencies with three request types: \emph{accesses} to existing nodes in the list,
\emph{insertions} of new nodes as well as \emph{deletions} of existing nodes in the list.
Upon receiving an insertion request, the node reveals its dependencies with the nodes that are already in the list.
The revealed dependencies are permanent, and must be obeyed until a node is deleted.
Note that is possible that a node will have dependencies with nodes that will be inserted later, but this information is hidden until then.

\medskip

Inserting a new node may require extensive rearrangements to even meet the precedence constraints.
In Figure~\ref{fig:example4_1} we present an example of such a costly rearrangement required to insert a single node.

\begin{figure}[h]
    \label{ex:general}
    \label{fig:example4_1}
    \center
    \includegraphics[width=0.8\textwidth]{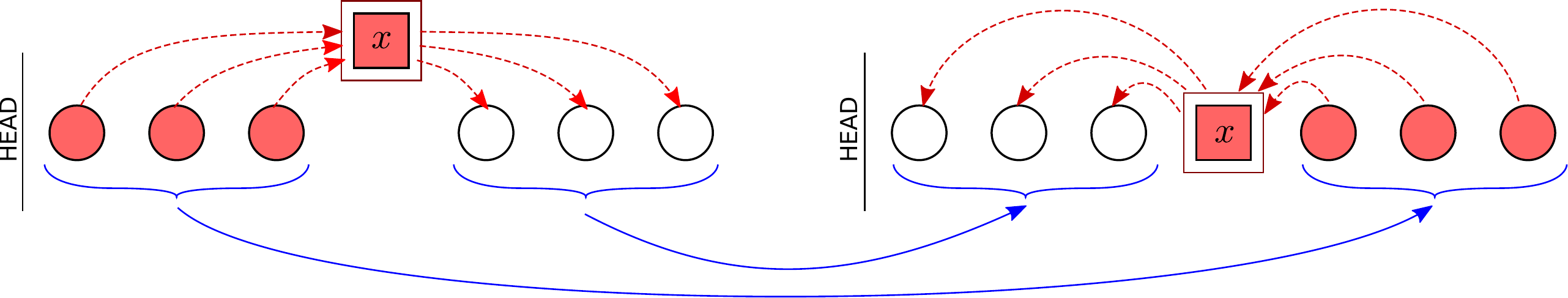}
    \caption{
    Consider a set of $n$ nodes in a list, with no dependencies between them. 
Then, consider an insertion of a node $x$, which reveals the depicted dependencies. Here, $x$ has
$\frac{n}{2}$ dependencies from the first half list (the red nodes) to $x$, and
$\frac{n}{2}$ dependencies from $x$ to the nodes of the second half of the list (the white nodes).
To insert the item without violating the dependencies, the nodes of the first half of the list must move ahead of the nodes of the second half, thus the cost of rearrangement is $(\frac{n}{2})^2 $.
    }
    \label{fig:example4_1}
\end{figure}

\vspace{-0.2cm}
\paragraph*{Constant-competitiveness with additional assumptions.}
To remedy incurring such a large cost at insertion, we discuss and rationalize the assumptions that guarantee a~feasible place to insert any node without rearrangements.
With these assumptions, we retain constant-competiveness of our algorithms \DET and \RAND.
The Assumption~\ref{asm:reconstruct} concerns the cost of each insertion, and the Assumption~\ref{asm:transitive} concerns the structure of the dependency graph.

\begin{assumption}
    \label{asm:reconstruct}
    Any algorithm incurs cost $n$ inserting a node at any position in the list.
\end{assumption}
\vspace{-0.6cm}
\paragraph*{Comment on Assumption~\ref{asm:reconstruct}.}
This assumption is common in the literature, as also in the classic list access problem, insertion require similar assumptions: 
nodes are inserted at the last position, and nodes are accessed before deletion~\cite{sleator1985amortized}.
We rely on an equivalent assumption in the proposed solution.
Assumption~\ref{asm:reconstruct} may be justified by a need to check the incoming node against possible
dependencies with existing nodes.
Upon receiving an insertion request of a node $\sigma_t$, an algorithm first determines the change in the dependency graph,
by comparing $\sigma_t$ with all nodes in the list.
The cost of determining dependencies equals the number of nodes in the list at the time of $\sigma_t$'s arrival.
Then, any algorithm inserts the node $\sigma_t$ at any position of its choice that respects the dependency DAG $G$, \emph{without incurring further cost}.

\medskip

\begin{assumption}
    \label{asm:transitive}
    The dependency DAG \DAG is \emph{transitive}.
\end{assumption}
\vspace{-0.6cm}
\paragraph*{Comment on Assumption~\ref{asm:transitive}.}
Consider the universe of all nodes, including the ones that that may be inserted in the future.
Assume the transitivity of their precedence constraints, meaning that if a~node $v_1$ must be in front of $v_2$, and $v_2$ must be in front of $v_3$, then $v_1$ must be in front of $v_3$, \emph{even if $v_2$ is currently not present in the list}.

\medskip

\pagebreak
Consider the modification of the algorithms \DET and \RAND that upon receiving an insertion request, inserts the new node into an arbitrary feasible position (respecting the dependencies).
Upon receiving an access request, the behavior of the algorithms remain unchanged.
Upon receiving a~deletion request, the algorithms delete the node.
When \RAND inserts a node, it assigns a bit to it uniformly at random from $\{0,1\}$.

\medskip
\begin{restatable}{Theorem}{thmInsertions}
    For a sequence $\sigma$ consisting of accesses, insertions and deletions,
    under Assumptions~\ref{asm:reconstruct} and \ref{asm:transitive},
        the algorithms \DET and \RAND with insertions and deletions are strictly 4-competitive.
    \label{thm:insertions-constant}
\end{restatable}

With Assumption~\ref{asm:transitive}, it is always possible to insert a node without rearranging the list --- a position that satisfies all constraints already exists in every configuration.

\begin{proof}
    First, consider the algorithm \DET.
    Fix a sequence of requests $\sigma$ of length $m$.
    We compare the costs of \DET and an optimal offline algorithm \OPT on $\sigma$ using the potential function $\Phi$, defined as \emph{twice the number of inversions}.
    Let $C_{\DET}(t)$ and $C_{\OPT}(t)$ denote the cost incurred at time $t$ by \DET and \OPT respectively.

    Consider any insertion request.
    Due to the transitivity assumption (Assumption~\ref{asm:transitive}) for the DAG, we know that there exists a position in any list where the incoming node can be inserted without node rearrangements.
    Recall that by our model assumption (Assumption~\ref{asm:reconstruct}), upon receiving an insertion request, any algorithm needs to check all the nodes in the list for possible dependencies, and then it can insert the new node in a position of its choice for free.
    Thus, for an insertion request, both \DET and \OPT pay exactly $C_{\OPT}(t)=C_{\DET}(t)=n$.
    Note that \DET and \OPT might place the node at different positions, thus we account for at most $n$ inversions created, and consequently for the change in potential we have $\Delta \Phi(t) \leq 2 n $.
    This gives us
        $C_{\DET}(t)+\Delta \Phi(t) \leq 3 n \leq 3\cdot C_{\OPT}(t)$.

    In case of a deletion request, exactly $\ell$ inversions are removed, and since $C_{\OPT}(t)=k$ and $C_{\DET}(t)=k+\ell$, we have
        $C_{\DET}(t)+\Delta \Phi(t) \leq (k+\ell)-2 \ell \leq k \leq 1 \cdot C_{\OPT}(t)$.

    Similarly to the proof of Theorem~\ref{thm:follow-4-competitive}, the amortized cost for an access request is
    $C_{\DET}(t) + \Delta \Phi(t)  \le 4 \cdot C_{\OPT}(t)$.
    We sum up the amortized bounds for all requests of the sequence $\sigma$, of all types (access, insertion, deletion),  and we conclude that \DET is 4-competitive.

    \medskip

    Second, consider the algorithm \RAND. The argument is similar to the proof for \DET.
    Fix a sequence of requests $\sigma$ of length $m$.
    We use a potential function $\Phi$ where, we distinguish between two types of inversions; we say that an inversion of the ordered pair $(y,\sigma_t)$ is \emph{type $i$} if $\bit(\sigma_t)=i$ for $i\in\{0,1\}$.
    Precisely, we use the potential function
    \begin{equation*}
        \Phi= \frac{5}{2}\cdot I_0 + \frac{7}{2}\cdot I_1,
    \end{equation*}
    where $I_0$ is the number of inversions of type $0$ and $I_1$ is the number of inversions of type $1$.

    Consider any insertion request.
    Due to the transitivity assumption (Assumption~\ref{asm:transitive}) for the DAG, we know that there exists a position in a list where the incoming node can be inserted without node rearrangements.
    Recall that by model assumption (Assumption~\ref{asm:reconstruct}), upon receiving an insertion request, any algorithm needs to check all the nodes in the list for possible dependencies, and then it inserts the node in a position of its choice for free. Also, for insertion in \RAND, we assign a bit from $\{0,1\}$ to a new node uniformly at random.
    Thus, for an insertion request, both \RAND and \OPT pay exactly $C_{\OPT}(t)=C_{\RAND}(t)=n$.
    Note that \RAND and \OPT may have placed the node at different positions, thus we account for at most $n$ inversions created.
    By Observation~\ref{obs:independent-bits}, the bits of the nodes are independent and equally likely to be 0 or 1, and we note that this property is true with insertions as well.
    Thus, we bound the expected change in potential by $E[\Delta \Phi(t)] \leq (\frac{1}{2}\cdot \frac{5}{2} + \frac{1}{2}\cdot \frac{7}{2})\cdot n = 3 n $.
    This gives us
        $E[C_{\RAND}(t)+\Delta \Phi(t)] \leq 4 n \leq 4\cdot C_{\OPT}(t)$.

        In case of a deletion request, exactly $\ell$ inversions are removed, and since $C_{\OPT}(t)=k$ and $C_{\DET}(t)=k+\ell$, we have
        $E[C_{\RAND}(t)+\Delta \Phi(t)] \leq (k+\ell)-\frac{5}{2}\cdot \ell \leq k \leq 1 \cdot C_{\OPT}(t)$,
        where the first inequality holds since the deletion of type 0 inversion decreases the potential the least.
    Finally, we sum up the amortized bounds for all requests of the sequence $\sigma$, of all types (access, insertion, deletion), and we conclude that \RAND is 4-competitive.
\end{proof}

\vspace{-0.3cm}
\paragraph*{Packet classification and transitivity.}
Finally, we note that the packet classification rule dependencies may not have the transitivity property (for an example, see Table~\ref{fig:example4_2}).
To deal with this issue without sacrificing competitiveness, either an amortized analysis is needed, or a solution that is related to the nature of the rules, such as rule preprocessing.

\begin{table}[!h]
    \centering
\begin{tabular}{lllllll}
\hline
\textbf{N} & \textbf{Proto} & \textbf{Src IP} & \textbf{Dst IP} & \textbf{Src Port} & \textbf{Dst Port} & \textbf{Action} \\ \hline
1 & TCP & 10.1.1.1 & 20.1.1.1 & ANY & 80 & ACCEPT \\
2 & TCP & 10.1.1.2 & 20.1.1.1 & ANY & 80 & ACCEPT \\
3 & TCP & 10.1.1.3 & 20.1.1.1 & ANY & 80 & ACCEPT \\
x & TCP & 10.1.1.0/24 & 20.1.1.1 & ANY & ANY & DENY \\
4 & TCP & 0.0.0.0/0 & 0.0.0.0/0 & ANY & 445 & ACCEPT \\
5 & TCP & 0.0.0.0/0 & 0.0.0.0/0 & ANY & 17 & ACCEPT \\
6 & TCP & 0.0.0.0/0 & 0.0.0.0/0 & ANY & 18 & ACCEPT \\
\hline
\end{tabular}
\caption{
Consider a set of $n$ nodes with no dependencies between them. 
Then, consider an insertion of a rule $x$, which reveals the dependencies with all existing rules.
The insertion of $x$ enforces the order between all existing rules.
The dependencies imposed by the rules are not transitive. With~the transitivity assumption, the rules (4-6) would have dependencies to rules (1-3), and the rule $x$ would be inserted between them without rearrangements. The situation is illustrated in Figure~\ref{fig:example4_1}.}
   \label{fig:example4_2}

\end{table}

\section{Conclusions and Future Directions}
\label{sec:conclusions}

We introduced a generalization of the classic list access problem with dependencies, and showed that
this generalization still admits constant-competitive online algorithms.
We see two interesting avenues for future research. On the theoretical front, it would be interesting
to have tight bounds on the competitive ratio, in both the deterministic and the randomized setting, and to analyze insertions and deletions without the transitivity assumption.
On the applied front, it would be interesting to engineer our algorithms towards enhancing existing
software-based packet classifiers.

\paragraph*{Acknowledgements.}
Research supported by the Austrian Science Fund (FWF) and the Hungarian National Research, Development and Innovation Office NKFIH, I 5025-N, 2020-2024.

\raggedbottom
\pagebreak
\bibliographystyle{plainurl}
\bibliography{firewall}

\begin{thebibliography}{10}

\bibitem{Acharya2007}
Subrata Acharya, Bryan~N. Mills, Mehmud Abliz, Taieb Znati, Jia Wang, Zihui Ge,
  and Albert~G. Greenberg.
\newblock {OPTWALL:} {A} hierarchical traffic-aware firewall.
\newblock In {\em Proceedings of the Network and Distributed System Security
  Symposium, NDSS}, 2007.

\bibitem{Albers03}
Susanne Albers.
\newblock Online algorithms: a survey.
\newblock {\em Math. Program.}, 97(1-2):3--26, 2003.

\bibitem{Albers2020}
Susanne Albers and Maximilian Janke.
\newblock New bounds for randomized list update in the paid exchange model.
\newblock In {\em Proceedings of the International Symposium on Theoretical
  Aspects of Computer Science, {STACS}}, volume 154, pages 1--17, 2020.

\bibitem{albers1997revisiting}
Susanne Albers and Michael Mitzenmacher.
\newblock Revisiting the counter algorithms for list update.
\newblock {\em Information processing letters}, 64(3):155--160, 1997.

\bibitem{Azar2002}
Yossi Azar and Leah Epstein.
\newblock On-line scheduling with precedence constraints.
\newblock {\em Discret. Appl. Math.}, 119(1-2):169--180, 2002.

\bibitem{Bienkowski2017}
Marcin Bienkowski, Jan Marcinkowski, Maciej Pacut, Stefan Schmid, and
  Aleksandra Spyra.
\newblock Online tree caching.
\newblock In {\em Proceedings of the {ACM} Symposium on Parallelism in
  Algorithms and Architectures, {SPAA}}, pages 329--338, 2017.

\bibitem{Borodin1998}
Allan Borodin and Ran {El-Yaniv}.
\newblock {\em Online Computation and Competitive Analysis}.
\newblock Cambridge University Press, 1998.

\bibitem{Eppstein2001}
David Eppstein and S.~Muthukrishnan.
\newblock Internet packet filter management and rectangle geometry.
\newblock In {\em Proceedings of the Twelfth Annual Symposium on Discrete
  Algorithms, {SODA}}, pages 827--835, 2001.

\bibitem{garefalakis1997new}
Theodoulos Garefalakis.
\newblock A new family of randomized algorithms for list accessing.
\newblock In {\em European Symposium on Algorithms}, pages 200--216. Springer,
  1997.

\bibitem{gupta2001algorithms}
Pankaj Gupta and Nick McKeown.
\newblock Algorithms for packet classification.
\newblock {\em IEEE Network}, 15(2):24--32, 2001.

\bibitem{Hamed2006}
Hazem~H. Hamed and Ehab Al{-}Shaer.
\newblock Dynamic rule-ordering optimization for high-speed firewall filtering.
\newblock In {\em Proceedings of the {ACM} Symposium on Information, Computer
  and Communications Security, {ASIACCS}}, pages 332--342, 2006.

\bibitem{Kamali2013}
Shahin Kamali, Susana Ladra, Alejandro L{\'{o}}pez{-}Ortiz, and Diego Seco.
\newblock Context-based algorithms for the list-update problem under
  alternative cost models.
\newblock In {\em Proceedings of the Data Compression Conference, {DCC}}, pages
  361--370, 2013.

\bibitem{KamaliSurvey2013}
Shahin Kamali and Alejandro L{\'{o}}pez{-}Ortiz.
\newblock A survey of algorithms and models for list update.
\newblock volume 8066 of {\em Lecture Notes in Computer Science}, pages
  251--266. Springer, 2013.

\bibitem{Lopez-OrtizRR20}
Alejandro L{\'{o}}pez{-}Ortiz, Marc~P. Renault, and Adi Ros{\'{e}}n.
\newblock Paid exchanges are worth the price.
\newblock {\em Theor. Comput. Sci.}, 824-825:1--10, 2020.

\bibitem{10.1145/3371927.3371929}
Sebastiano Miano, Matteo Bertrone, Fulvio Risso, Mauricio~V\'{a}squez Bernal,
  Yunsong Lu, and Jianwen Pi.
\newblock Securing {Linux} with a faster and scalable iptables.
\newblock {\em SIGCOMM Comput. Commun. Rev.}, 49(3):2–17, 2019.

\bibitem{Munro2000}
J.~Ian Munro.
\newblock On the competitiveness of linear search.
\newblock In {\em Proceedings of the European Symposium, {ESA}}, volume 1879,
  pages 338--345, 2000.

\bibitem{openflow}
{ONF}.
\newblock Openflow reference release.
\newblock \url{https://github.com/mininet/openflow}, 2013.

\bibitem{Reingold1994}
Nick Reingold, Jeffery~R. Westbrook, and Daniel~Dominic Sleator.
\newblock Randomized competitive algorithms for the list update problem.
\newblock {\em Algorithmica}, 11(1):15--32, 1994.

\bibitem{sleator1985amortized}
Daniel~D. Sleator and Robert~E. Tarjan.
\newblock Amortized efficiency of list update and paging rules.
\newblock {\em Commun. ACM}, 28(2):202--208, February 1985.

\bibitem{Srinivasan1999}
Venkatachary Srinivasan, Subhash Suri, and George Varghese.
\newblock Packet classification using tuple space search.
\newblock In {\em Proceedings of the {ACM} {SIGCOMM}}, pages 135--146, 1999.

\end{thebibliography}

\end{document}